\documentclass[a4paper,UKenglish]{lipics-v2016}

\usepackage{microtype}
\usepackage{graphicx,amssymb,amsmath}
\usepackage{xcolor}
\usepackage{graphicx}
\usepackage{caption}
\usepackage{enumerate}
\usepackage{enumitem}
\usepackage{wrapfig}

\newcommand{\REMOVE}[1]{}

\newcommand{\eps}{\varepsilon}
\newcommand{\dist}{{\rm dist}}

\title{Reconstruction of Weakly Simple Polygons from their Edges}
\titlerunning{Reconstruction of Weakly Simple Polygons from their Edges} 

\author[1]{Hugo A. Akitaya}
\author[1,2]{Csaba D. T\'oth}
\affil[1]{Tufts University, Medford, MA, USA\\
  \texttt{hugo.alves\_akitaya@tufts.edu}}
\affil[2]{California State University Northridge, Los Angeles, CA, USA\\
  \texttt{cdtoth@acm.org}}
\authorrunning{H.\,A. Akitaya and C.\,D. T\'oth} 

\Copyright{Hugo A. Akitaya and Csaba D. T\'oth}


\EventEditors{Seok-Hee Hong}
\EventNoEds{1}
\EventLongTitle{27th International Symposium Algorithms and Computation
	(ISAAC 2016)}
\EventShortTitle{ISAAC 2016}
\EventAcronym{ISAAC}
\EventYear{2016}
\EventDate{December 12--14, 2016}
\EventLocation{Sydney, Australia}
\EventLogo{}
\SeriesVolume{64}
\ArticleNo{143}

\begin{document}

\maketitle

\begin{abstract}
Given $n$ line segments in the plane, do they form the edge set of a \emph{weakly simple polygon}; that is, can the segment endpoints be perturbed by at most $\varepsilon$, for any $\varepsilon>0$, to obtain a simple polygon? While the analogous question for \emph{simple polygons} can easily be answered in $O(n\log n)$ time, we show that it is NP-complete for weakly simple polygons. We give $O(n)$-time algorithms in two special cases: when all segments are collinear, or the segment endpoints are in general position. These results extend to the variant in which the segments are \emph{directed}, and the counterclockwise traversal of a polygon should follow the orientation.

We study related problems for the case that the union of the $n$ input segments is connected.
(i) If each segment can be subdivided into several segments, find the minimum number of subdivision points to form a weakly simple polygon.
(ii) If new line segments can be added, find the minimum total length of new segments that creates a weakly simple polygon. We give worst-case upper and lower bounds for both problems.
 \end{abstract}

\section{Introduction}
\label{sec:intro}

In the design and analysis of geometric algorithms, the input is often assumed to be in general position. This is justified from the theoretical point of view: degenerate cases can typically be handled without increasing the computational complexity, or symbolic perturbation schemes can reduce any input to one in general position~\cite{Dutch}. In this paper, we present a geometric problem about simple polygons in the plane, which has a straightforward solution if the input is in general position, but is NP-complete otherwise.

Suppose we are given $n$ line segments in the plane. It is easy to decide in $O(n\log n)$ time whether they form a simple polygon by detecting intersections in a line sweep: if the segments are disjoint apart from common endpoints, then they form a plane graph, and a simple traversal can determine whether the graph is a cycle. If the input segments overlap, more than two segments have a common endpoint, or some segment endpoints lie in the interior of another segment, then they definitely do not form a simple polygon, but they might still be perturbed into a simple polygon (i.e., they form a weakly simple polygon). We study
the decision problem for weakly simple polygons in this paper.

\smallskip\noindent{\bf Organization and Results.}
We start with necessary definitions, and formulate the problem of reconstructing a weakly simple polygons from a set of edges (Section~\ref{sec:def}).
We present polynomial-time algorithms when the given segments form a geometric graph or are collinear (Sections~\ref{sec:special}).
The problem in general, however, is strongly NP-hard by a reduction from \textsc{Planar-Monotone-3SAT} (Section~\ref{sec:hard}). Nevertheless,
every set of noncrossing line segments in the plane can be turned into the edge set of a weakly simple polygon by (i) subdividing the edges into several edges, or (ii) inserting new edges.
In Sections~\ref{sec:bounds} we show that if $G=(V,E)$ is Eulerian, the edges can be subdivided $O(n)$ times to obtain a weakly simple Euler tour.
We also show that inserting new edges of total length at most $3\|E\|$  is always sufficient and sometimes necessary to create a weakly simple Euler tour. We conclude with future directions (Section~\ref{sec:conclusions}).
Omitted proofs are available in the Appendix.

\smallskip\noindent{\bf Related Work.}
Reconstruction of simple polygons from partial information (such as vertices, visibility graphs, visibility angles, cross sections) has been studied for decades~\cite{BGS06,BDS11,DMW11,FW90,OR88}. For example, an orthogonal simple polygon can be uniquely reconstruction from its vertices~\cite{OR88}, but if the edges have 3 or more directions, the problem becomes NP-hard~\cite{FW90}. For a simple polygon, the set of all edges (studied in this paper) gives complete information: the cyclic order of the edges is easy to recover. In contrast, a set of edges may correspond to exponentially many weakly simple polygons, and the reconstruction problem becomes nontrivial.
The problems considered in Section~\ref{sec:bounds} are closely related to geometric graph augmentation and subgraph problems: (i) Can a given plane straight-line graph be augmented with new edges into a simple polygon, a Hamiltonian plane graph, or a 2-connected plane graph~\cite{HT03,Rap89,Tot12,UW92}? (ii) Does a given a geometric graph contain certain noncrossing subgraphs (e.g., spanning trees or perfect matchings)~\cite{JW93}?

\section{Preliminaries}
\label{sec:def}

A \emph{polygon} $P=(p_0,\ldots ,p_{n-1})$ is a cyclic sequence of points in the plane  (\emph{vertices}), where every two consecutive vertices are connected by a line segment (\emph{edge}). The cycle of edges can be parameterized by a piecewise linear curve $\gamma: \mathbb{S}^1\rightarrow \mathbb{R}^2$. Polygon $P$ is \emph{simple} if $\gamma$ is a \emph{Jordan curve} (i.e., $\gamma$ is injective); equivalently, if $(p_0,\ldots ,p_{n-1})$ is the plane embedding of a Hamiltonian cycle. Polygon $P$ is \emph{weakly simple} if, for every $\eps>0$, the vertices $p_i$ can be perturbed to points $p_i'$, $\|p_i p_i'\|<\eps$, such that $P'=(p_0',\ldots , p_{n-1}')$ is a simple polygon. The function $\|.\|$ denotes the Euclidean length of a line segment. Equivalently, a polygon given by  $\gamma$ is weakly simple if it can be perturbed to Jordan curve $\gamma':\mathbb{S}^1\rightarrow \mathbb{R}^2$ such that the Fr\'echet distance of the two curves is bounded by $\eps$ (i.e., $\dist_F(\gamma,\gamma')<\eps$)~\cite{CEX15}. We can test whether a polygon $P=(p_0,\ldots ,p_{n-1})$, is simple or weakly simple, respectively, in $O(n)$ time~\cite{Cha91} and $O(n\log n)$ time~\cite{AAET16}.

	\begin{figure}[h]
		\centering
		\def\svgwidth{0.7\columnwidth}
\begingroup%
  \makeatletter%
  \providecommand\color[2][]{%
    \errmessage{(Inkscape) Color is used for the text in Inkscape, but the package 'color.sty' is not loaded}%
    \renewcommand\color[2][]{}%
  }%
  \providecommand\transparent[1]{%
    \errmessage{(Inkscape) Transparency is used (non-zero) for the text in Inkscape, but the package 'transparent.sty' is not loaded}%
    \renewcommand\transparent[1]{}%
  }%
  \providecommand\rotatebox[2]{#2}%
  \ifx\svgwidth\undefined%
    \setlength{\unitlength}{320.60002124bp}%
    \ifx\svgscale\undefined%
      \relax%
    \else%
      \setlength{\unitlength}{\unitlength * \real{\svgscale}}%
    \fi%
  \else%
    \setlength{\unitlength}{\svgwidth}%
  \fi%
  \global\let\svgwidth\undefined%
  \global\let\svgscale\undefined%
  \makeatother%
  \begin{picture}(1,0.18958315)%
    \put(0,0){\includegraphics[width=\unitlength,page=1]{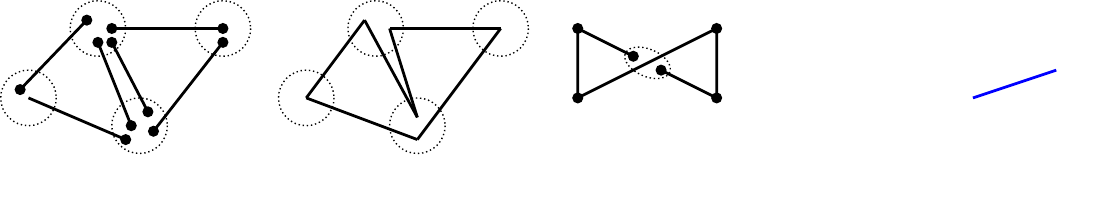}}%
    \put(0.10224691,0.00545448){\color[rgb]{0,0,0}\makebox(0,0)[lb]{\smash{(a)}}}%
    \put(0.34558285,0.00541566){\color[rgb]{0,0,0}\makebox(0,0)[lb]{\smash{(b)}}}%
    \put(0.56145725,0.00493644){\color[rgb]{0,0,0}\makebox(0,0)[lb]{\smash{(c)}}}%
    \put(0,0){\includegraphics[width=\unitlength,page=2]{examples.pdf}}%
    \put(0.7417634,0.00493641){\color[rgb]{0,0,0}\makebox(0,0)[lb]{\smash{(d)}}}%
    \put(0,0){\includegraphics[width=\unitlength,page=3]{examples.pdf}}%
    \put(0.91643588,0.00493641){\color[rgb]{0,0,0}\makebox(0,0)[lb]{\smash{(e)}}}%
    \put(0,0){\includegraphics[width=\unitlength,page=4]{examples.pdf}}%
  \end{picture}%
\endgroup%
	\caption{
    (a) A multi-set of line segments. Circles indicate common segment endpoints.
    (b) A weakly simple Euler tour.
    (c) An Eulerian graph that has no weakly simple Euler tour.
    An edge subdivision (d) or the insertion of two edges (e), yields a weakly simple Euler tour.}
	\label{fig:intro}
	\end{figure}

We define the \textsc{WeaklySimplePolygonReconstruction} (WSPR) problem as the following decision problem: Given a multiset $E$ of line segments in $\mathbb{R}^2$, does there exist a weakly simple polygon $P$ whose edge multiset is $E$? For a multiset $E$ of \emph{directed} segments, we also define \textsc{Directed-WSPR} that asks whether there exists a weakly simple polygon $P=(p_0,\ldots , p_{n-1})$ such that $\{p_ip_{i+1\mod n}: 0\leq i\leq n-1\}=E$.
In both undirected and directed variants, we represent the input segments as a straight-line multigraph $G=(V,E)$, where $V$ is the set of all segment endpoints. Note that $G$ may have overlapping edges, and an edge may pass through vertices, so it need not be a \emph{geometric graph}.

\smallskip\noindent{\bf Two Necessary Conditions.}
Two line segments \emph{cross} if they share exactly one interior point.
If the multiset of segments $E$ forms a weakly simple polygon, then no two segment cross.
This condition can be easily tested in $O(|E|\log |E|)$ time by a line sweep.

If there is a  weakly simple polygon $P=(p_0,\ldots ,p_{n-1})$ with edge set $E$, then $P$ is an Euler tour of the graph $G=(V,E)$. (However, an Euler tour need not be weakly simple; see Fig.~\ref{fig:intro}(b)).
A graph is Eulerian if and only if it is connected and every vertex has even degree.
A \emph{simple} (undirected) plane graph $G$ is Eulerian if and only if its dual graph is bipartite.
This result extends to plane \emph{multi}graphs when an edge of multiplicity $k$ is embedded as $k$ interior-disjoint Jordan arcs, that enclose $k-1$ faces. 
A directed graph is Eulerian if and only if all vertices are part of the same strongly connected component and if, for each vertex, the in-degree equals the out-degree.

\section{Special Cases}
\label{sec:special}

We show that both WSPR and Directed-WSPR admit polynomial-time algorithms
in the special cases that (i) $G=(V,E)$ is a \emph{simple geometric graph}, that is, no two edges overlap,
and no vertex lies in the interior of an edge, and (ii)  all edges in $G=(V,E)$ are collinear. We assume that
$G$ satisfies both necessary conditions.

\subsection{Geometric Graphs}
\label{sec:non-overlap}

Note that in an Eulerian geometric graph the boundary of each face is a weakly simple circuit, where repeated vertices are
possible, but there are no repeated edges.
The following is a modified version of Hierholzer's algorithm~\cite{Hie73}.
It computes a weakly simple Euler tour $P$ in the Eulerian graph $G$,
or reports that no such tour exists.

\smallskip
\noindent\textbf{Algorithm A $(G)$}
\begin{enumerate}\itemsep -1pt
	\item
	2-color the faces of $G$ white and gray so that the outer face is white;
	and create a list $L$ of circuits on the boundaries of the gray faces.
	\item
	If $G$ is directed and the edges around a gray face do not form a directed circuit or if there exist both clockwise (cw) and counterclockwise (ccw) circuits in $L$, report that $G$ has no weakly
	simple Euler tour.
	\item
	Choose an arbitrary circuit in $L$, remove it from $L$ and call it $P$.
	\item
	While there is a  circuit in $L$, do:
	\begin{enumerate}
		\item[4.1]
		Find two consecutive edges, $(u,v)$ and $(v,w)$, along a white face
		such that $(u,v)\in P$ and $(v,w)\in C$ for some $C\in L$.
		\item[4.2]
		Remove $C$ from the list $L$, and merge $C$ and $P$ by traversing $C$ starting with the edge $(v,w)$ followed by the traversal of $P$ that ends with the edge $(u,v)$; see Figure~\ref{fig:geometric}(a).
	\end{enumerate}
	\item
	Return $P$.
\end{enumerate}

\begin{figure}[h]
	\centering
	\def\svgwidth{0.8\columnwidth}
\begingroup%
  \makeatletter%
  \providecommand\color[2][]{%
    \errmessage{(Inkscape) Color is used for the text in Inkscape, but the package 'color.sty' is not loaded}%
    \renewcommand\color[2][]{}%
  }%
  \providecommand\transparent[1]{%
    \errmessage{(Inkscape) Transparency is used (non-zero) for the text in Inkscape, but the package 'transparent.sty' is not loaded}%
    \renewcommand\transparent[1]{}%
  }%
  \providecommand\rotatebox[2]{#2}%
  \ifx\svgwidth\undefined%
    \setlength{\unitlength}{398.63709938bp}%
    \ifx\svgscale\undefined%
      \relax%
    \else%
      \setlength{\unitlength}{\unitlength * \real{\svgscale}}%
    \fi%
  \else%
    \setlength{\unitlength}{\svgwidth}%
  \fi%
  \global\let\svgwidth\undefined%
  \global\let\svgscale\undefined%
  \makeatother%
  \begin{picture}(1,0.25778112)%
    \put(0,0){\includegraphics[width=\unitlength,page=1]{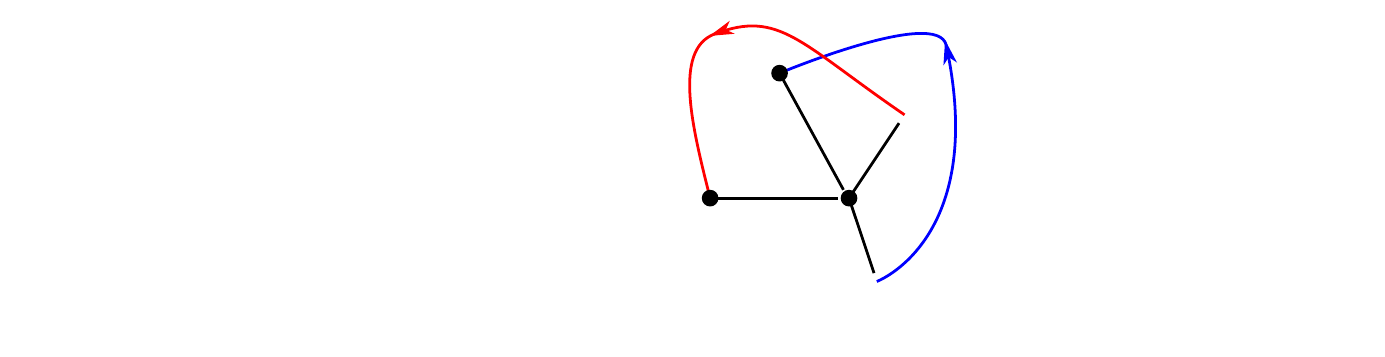}}%
    \put(0.53507714,0.20073124){\color[rgb]{0,0,0}\makebox(0,0)[lb]{\smash{$a$}}}%
    \put(0,0){\includegraphics[width=\unitlength,page=2]{geometricv2.pdf}}%
    \put(0.66503313,0.16900443){\color[rgb]{0,0,0}\makebox(0,0)[lb]{\smash{$b$}}}%
    \put(0.60503215,0.14874008){\color[rgb]{0,0,0}\rotatebox{-4.03533038}{\makebox(0,0)[lb]{\smash{...}}}}%
    \put(0.58521501,0.09390416){\color[rgb]{0,0,0}\rotatebox{-30.36505974}{\makebox(0,0)[lb]{\smash{...}}}}%
    \put(0.63738905,0.11964697){\color[rgb]{0,0,0}\rotatebox{-95.53446767}{\makebox(0,0)[lb]{\smash{...}}}}%
    \put(0.47246663,0.18747469){\color[rgb]{0,0,0}\makebox(0,0)[lb]{\smash{$\pi_1$}}}%
    \put(0.68982335,0.21724399){\color[rgb]{0,0,0}\makebox(0,0)[lb]{\smash{$\pi_2$}}}%
    \put(0.48681599,0.10553316){\color[rgb]{0,0,0}\makebox(0,0)[lb]{\smash{$c$}}}%
    \put(0.60463471,0.04813375){\color[rgb]{0,0,0}\makebox(0,0)[lb]{\smash{$d$}}}%
    \put(0.58039706,0.12513392){\color[rgb]{0,0,0}\makebox(0,0)[lb]{\smash{$v$}}}%
    \put(0,0){\includegraphics[width=\unitlength,page=3]{geometricv2.pdf}}%
    \put(0.79596606,0.20073124){\color[rgb]{0,0,0}\makebox(0,0)[lb]{\smash{$a'$}}}%
    \put(0,0){\includegraphics[width=\unitlength,page=4]{geometricv2.pdf}}%
    \put(0.92592205,0.16900442){\color[rgb]{0,0,0}\makebox(0,0)[lb]{\smash{$b'$}}}%
    \put(0.72934186,0.18747469){\color[rgb]{0,0,0}\makebox(0,0)[lb]{\smash{$\pi_1'$}}}%
    \put(0.95344699,0.21564945){\color[rgb]{0,0,0}\makebox(0,0)[lb]{\smash{$\pi_2'$}}}%
    \put(0.74770493,0.10553316){\color[rgb]{0,0,0}\makebox(0,0)[lb]{\smash{$c'$}}}%
    \put(0.86552362,0.04813374){\color[rgb]{0,0,0}\makebox(0,0)[lb]{\smash{$d'$}}}%
    \put(0.83727233,0.12112025){\color[rgb]{0,0,0}\makebox(0,0)[lb]{\smash{$v'$}}}%
    \put(0,0){\includegraphics[width=\unitlength,page=5]{geometricv2.pdf}}%
    \put(0.82924498,0.07496297){\color[rgb]{0,0,0}\makebox(0,0)[lb]{\smash{$v''$}}}%
    \put(0.75988162,0.18887415){\color[rgb]{0,0,0}\makebox(0,0)[lb]{\smash{}}}%
    \put(0,0){\includegraphics[width=\unitlength,page=6]{geometricv2.pdf}}%
    \put(0.14520649,0.14334343){\color[rgb]{0,0,0}\makebox(0,0)[lb]{\smash{$v$}}}%
    \put(0,0){\includegraphics[width=\unitlength,page=7]{geometricv2.pdf}}%
    \put(0.36468932,0.13050753){\color[rgb]{0,0,0}\makebox(0,0)[lb]{\smash{$v$}}}%
    \put(0,0){\includegraphics[width=\unitlength,page=8]{geometricv2.pdf}}%
    \put(0.12377676,0.2151914){\color[rgb]{0,0,0}\makebox(0,0)[lb]{\smash{$C$}}}%
    \put(0.07821103,0.08297638){\color[rgb]{0,0,0}\makebox(0,0)[lb]{\smash{$P$}}}%
    \put(0,0){\includegraphics[width=\unitlength,page=9]{geometricv2.pdf}}%
    \put(-0.00069265,0.1555017){\color[rgb]{0,0,0}\makebox(0,0)[lb]{\smash{$u$}}}%
    \put(0.04873807,0.24337296){\color[rgb]{0,0,0}\makebox(0,0)[lb]{\smash{$w$}}}%
    \put(0,0){\includegraphics[width=\unitlength,page=10]{geometricv2.pdf}}%
    \put(0.2181561,0.1555017){\color[rgb]{0,0,0}\makebox(0,0)[lb]{\smash{$u$}}}%
    \put(0.25542859,0.24337296){\color[rgb]{0,0,0}\makebox(0,0)[lb]{\smash{$w$}}}%
    \put(0,0){\includegraphics[width=\unitlength,page=11]{geometricv2.pdf}}%
    \put(0.28858606,0.09022196){\color[rgb]{0,0,0}\makebox(0,0)[lb]{\smash{$P$}}}%
    \put(0.17285651,0.00784197){\color[rgb]{0,0,0}\makebox(0,0)[lb]{\smash{(a)}}}%
    \put(0.69992543,0.00784197){\color[rgb]{0,0,0}\makebox(0,0)[lb]{\smash{(b)}}}%
  \end{picture}%
\endgroup%
	\caption{(a) Merging two cycles. The vertices circled by the dotted ellipse correspond to the same vertex $v$. (b) If $v$ has two consecutive incoming edges $ (a,v) $ and $ (c,v) $, $G$ does not admit a weakly simple Euler tour.}
	\label{fig:geometric}
\end{figure}

\begin{theorem}	\label{thm:GeomGraph}
A simple geometric graph $G=(V,E)$ admits a weakly simple Euler tour if and only if $G$ is Eulerian and, if $G$ is directed, the circular order of edges around each vertex alternates between incoming and outgoing.
A weakly simple Euler tour, if exists, can be computed in $O(|E|)$ time.
\end{theorem}

\begin{proof}
If $G=(V,E)$ is undirected, the algorithm construct an Euler tour $P$~\cite{Hie73}; and the tour is weakly simple by construction.
In the remainder of the proof, we consider a directed Eulerian geometric graph $G$.
	First, we show that if $G$ satisfies the conditions of Theorem~\ref{thm:GeomGraph}, then Algorithm~A returns a weakly simple Euler tour.
	If the circular order of edges around each vertex alternates between incoming and outgoing, then all edges on the boundary of a face of $G$
	have the same orientation (ccw or cw), and adjacent faces have opposite orientations.
	Without loss of generality, the edges on the boundaries of white (resp., gray) faces are oriented ccw (resp., cw). Hnece, the condition in step~2 of the algorithm is satisfied.
	
	We show that the Euler tour $P$ constructed by Algorithm A is weakly simple, that is, it can be perturbed into a simple polygon. Initially, each circuit $C=(p_0,\ldots , p_{k-1})$ in $L$ is the boundary of a gray face, and hence it is a simple polygon. Let $C'=(p_0',\ldots , p_{k-1}')$
	be perturbation obtained by moving each point $p_i$ to the interior of the face along an angle bisector of $\angle p_{i-1}p_ip_{i+1}$.
	Initially $P$ is a weakly simple polygon (one of the circuits).
	It is enough to show that Step~4.2 maintains a weakly simple polygon, that is, when we merge $P$ and a circuit $C$, their Jordan curve perturbations $P'$ and $C'$ can also be combined. Edges $(u,v)$ and $(v,w)$ are adjacent to a common white face $f_0$; they correspond to an edge $(p_u,p_v)$ in $P'$ and $(q_v,q_w)$ in $C'$, where both $p_v$ and $q_v$ lie in the $\eps$-neighborhood of $v$ in two different gray faces adjacent to $f_0$. We can modify $P'$ and $C'$ in the $\eps$-neighborhood of $v$, by removing a short Jordan arc from each and reconnecting them across the white face $f_0$ into a single Jordan curve.
	By induction, we can obtain a Jordan curve within $\eps$ Fr\`echet distance from the output polygon $P$. Hence, the algorithm returns a weakly simple Euler tour $P$.
	
	Now, we show that if $G$ has a vertex $v$ with two consecutive incoming (resp., outgoing) edges
	$(a,v)$ and $(c,v)$, then $G$ does not admit a weakly simple Euler tour.
	Suppose, for contradiction, that there exists a weakly simple Euler tour $P$.
	Since both $(a,v)$ and $(c,v)$ are directed into $v$, the tour $P$ contains
	edge-disjoint paths $(a,v,b)$ and $(c,v,d)$. Since $P$ is weakly simple,
	the circular order of these four edges incident to $v$ must be as shown in Figure~\ref{fig:geometric}(b).
    The polygon must contain edge-disjoint paths $\pi_1=(v,b,\ldots, c,v)$ and $\pi_2=(d,\ldots, a)$.
	The perturbation of $\pi_1$ is $\pi_1'=(v',b',\ldots , c',v'')$ where $v'\neq v''$.
	Note that $a$ and $d$ are on opposite sides of the cycle $\pi'\cup v'v''$.
    The perturbation of $\pi_2$, path $\pi_2'$, can intersect neither $\pi_1'$ nor $v'v''$, because
	$(a,v)$ and $(c ,v)$ are adjacent to the same face. Hence $P$ is not weakly simple.
	

	Finally, Algorithm A runs in $O(|E|)$ time.
	Step 1 and 2 can be done by traversing the dual graph of $G$.
	Step 4 executes $O(|E|)$ merges, each of which takes constant time.
\end{proof}



%
%
%

\begin{corollary}	\label{cor:GeomGraph}
A geometric multigraph $G=(V,E)$ admits a weakly simple Euler tour if and only if $G$ is Eulerian. A weakly simple Euler tour, if exists, can be computed in $O(|E|)$ time.
\end{corollary}
\begin{proof}
Replace every edge $e$ of multiplicity $k$ by $k$ edge-disjoint paths of length two whose interior points are close to the midpoint of $e$. We obtain a simple Eulerian geometric graph with $|V|+2|E|$ vertices. Theorem~\ref{thm:GeomGraph} completes the proof.
\end{proof}
\begin{remark}
In the case that $G=(V,E)$ is a directed multigraph, replace every directed edge $(u,v)$ of multiplicity $k$ by edge-disjoint paths $(u,w_i,v)$, with new (subdivision) vertices $w_i$, $i=1,\ldots k$, and denote by $G'$ the resulting simple directed graph. The alternating direction condition of Theorem~\ref{thm:GeomGraph} requires that the multiplicity of $(u,v)$ and $(v,u)$ differ by at most one. If their multiplicities differ by exactly one, then there is a unique way to interleave the replacement paths between $u$ and $v$. In fact, if any edge of $G$ has odd multiplicity, the alternating direction condition determines the cyclic order of all paths $(u,w_i,v)$, and we can apply Theorem~\ref{thm:GeomGraph} for $G'$. If, however, all edges of $G$ have even multiplicity, then there are two possibilities for the cyclic orders, both of which yield weakly simple Euler tours by Theorem~\ref{thm:GeomGraph}.
\end{remark}

\subsection{Collinear Line Segments}
\label{sec:collinear}
%

\begin{theorem}\label{thm:collinear}
	If all edges of a graph $G=(V,E)$ are collinear, then
	every Euler tour of $G$ is a weakly simple polygon.
\end{theorem}

\begin{wrapfigure}{r}{0.28\textwidth}
\vspace{-20pt}
  \begin{center}
    \includegraphics[width=0.25\textwidth]{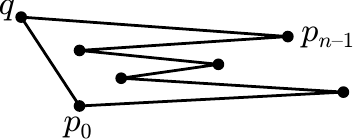}
    \caption{Every collinear Eulerian tour can be transformed in a $y$-monotonic simple polygon.}
	\label{fig:collinear}
  \end{center}
  \vspace{-20pt}
\end{wrapfigure}
\noindent{\bf Proof.}
	Without loss of generality, assume that all vertices are on the $x$-axis. Let $\eps>0$ be given.
	Let $P=(p_0,\ldots,p_{n-1})$ be an Euler tour of $G$, and let $p_0$ be a leftmost vertex.
	For each vertex $p_i$, $i\in\{0,\ldots,n-1\}$, create a point $p_i'$
	with $x(p_i')=x(p_i)$ and $y(p_i')=i\eps/(2n)$.
	The polygonal path $(p_0',\ldots,p_{n-1}')$ is strictly $y$-monotonic and
	therefore does not cross itself. The edge $(p_{n-1}',p_0')$ can be realized as
	a one-bend polyline $(p_{n-1}',q,p_0')$ with $q=(-\eps/2,\eps/2)$, which is outside
	of the axis-aligned bounding box of all other edges.
	Therefore $P'=(p_0',\ldots,p_{n-1}')$, illustrated in Figure~\ref{fig:collinear},
	is a simple polygon where $\dist_F(P,P')< \eps$.
\qed

\section{NP-Completeness}
\label{sec:hard}

In this section we analyze the general case of WSPR.
First we discuss the undirected case and then the direct version.

\begin{lemma}
	\label{lem:np}
	Both WSPR and Directed-WSCR are in NP.
\end{lemma}

\begin{proof}
    Given a polygon $P=(p_0,\ldots, p_{n-1})$ and a (directed) straight-line multigraph $G=(V,E)$,
	we can check whether $P$ is a (directed) Euler tour in $G$ in $O(|E|)$ time, and whether $P$ is weakly simple in $O(n\log n)$ time~\cite{AAET16}.
\end{proof}

We prove that both directed and undirected WSPC are strongly NP-hard in the general case by a reduction from \textsc{Planar-Monotone-3SAT}, which is strongly NP-hard~\cite{BK12}.
An instance of \textsc{Planar-Monotone-3SAT} consists of a plane bipartite graph $G_B$
whose partite sets are variables nodes and clauses nodes. The variable nodes are on the $x$-axis, the clause nodes are above or below the $x$-axis; each clause is adjacent to three variables. A clause is \emph{positive} if it lies above the $x$-axis, and \emph{negative} otherwise.
\textsc{Planar-Monotone-3SAT} asks if there is a binary assignment from $\{$\texttt{true},\texttt{false}$\}$ to the set of variables such that every positive clause is adjacent to at least one \texttt{true} variable and every negative clause is adjacent to at least one \texttt{false} variable.

\begin{lemma}
	\label{lem:undirected}
	Undirected WSPR is NP-hard.
\end{lemma}

\begin{proof}
	Given an instance of \textsc{Planar-Monotone-3SAT}, we build an instance of undirected WSPR as shown in Figure~\ref{fig:undirected-gadgets}(c).
	We split the construction into two basic gadgets.
	A variable and a clause gadget are shown in Figure~\ref{fig:undirected-gadgets}(a) and (b), respectively.
	The figure shows collinear edges distorted and colored for clarity.
	All vertices shown as small black disks are on the $x$-axis and vertices circled with a dotted ellipse represent the same graph vertex.
	\begin{figure}[t]
		\centering
			\def\svgwidth{0.95\linewidth}
\begingroup%
  \makeatletter%
  \providecommand\color[2][]{%
    \errmessage{(Inkscape) Color is used for the text in Inkscape, but the package 'color.sty' is not loaded}%
    \renewcommand\color[2][]{}%
  }%
  \providecommand\transparent[1]{%
    \errmessage{(Inkscape) Transparency is used (non-zero) for the text in Inkscape, but the package 'transparent.sty' is not loaded}%
    \renewcommand\transparent[1]{}%
  }%
  \providecommand\rotatebox[2]{#2}%
  \ifx\svgwidth\undefined%
    \setlength{\unitlength}{518.66511307bp}%
    \ifx\svgscale\undefined%
      \relax%
    \else%
      \setlength{\unitlength}{\unitlength * \real{\svgscale}}%
    \fi%
  \else%
    \setlength{\unitlength}{\svgwidth}%
  \fi%
  \global\let\svgwidth\undefined%
  \global\let\svgscale\undefined%
  \makeatother%
  \begin{picture}(1,0.26578548)%
    \put(0,0){\includegraphics[width=\unitlength,page=1]{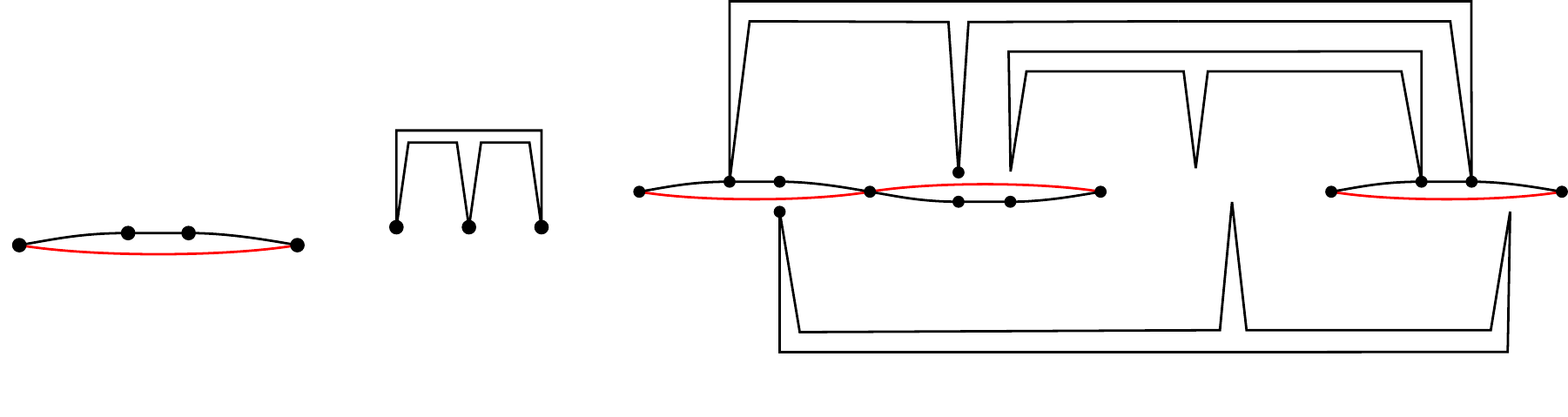}}%
    \put(0.08523663,0.06952186){\color[rgb]{0,0,0}\makebox(0,0)[lb]{\smash{(a)}}}%
    \put(0.28208981,0.04947039){\color[rgb]{0,0,0}\makebox(0,0)[lb]{\smash{(b)}}}%
    \put(0.69325454,0.00511082){\color[rgb]{0,0,0}\makebox(0,0)[lb]{\smash{(c)}}}%
    \put(0,0){\includegraphics[width=\unitlength,page=2]{undirected-gadgets-svg.pdf}}%
    \put(0.000044,0.12259203){\color[rgb]{0,0,0}\makebox(0,0)[lb]{\smash{$v_{i-1}$}}}%
    \put(0.07253779,0.13338896){\color[rgb]{0,0,0}\makebox(0,0)[lb]{\smash{$l_{i,a}$}}}%
    \put(0.11803739,0.13338896){\color[rgb]{0,0,0}\makebox(0,0)[lb]{\smash{$l_{i,b}$}}}%
    \put(0.18667695,0.12259202){\color[rgb]{0,0,0}\makebox(0,0)[lb]{\smash{$v_{i}$}}}%
    \put(0.24326271,0.0917436){\color[rgb]{0,0,0}\makebox(0,0)[lb]{\smash{$l_{i,a}$}}}%
    \put(0.28953535,0.0917436){\color[rgb]{0,0,0}\makebox(0,0)[lb]{\smash{$l_{j,a}$}}}%
    \put(0.3373504,0.0917436){\color[rgb]{0,0,0}\makebox(0,0)[lb]{\smash{$l_{k,a}$}}}%
  \end{picture}%
\endgroup%
		\caption{(a) Variable gadget, (b) clause gadget and (c) the reduction from \textsc{Planar-Monotone-3SAT} to undirected WSPR.}
		\label{fig:undirected-gadgets}
	\end{figure}
	First, place vertices $v_0,\ldots,v_n$ equally spaced on the variable line from left to right.
	The variable gadget corresponding to the $i$th variable consists of two collinear paths between $v_{i-1}$ and $v_i$, which are called \emph{red} and \emph{black} paths;
    see Figure~\ref{fig:undirected-gadgets}.
	The red path is a single edge $v_{i-1}v_i$; and the black path is made of $p+1$ edges where $p$ is the degree of the $i$-th variable in the bipartite graph $G_B$.
	We assign a vertex in the interior of this path to each edge connected to the variable, naming the vertex $l_{i,a}$ for the edge connecting the $i$-th variable to the $a$-th clause.
	We call such vertices \emph{literal} vertices.
	The clause gadget is composed of 9 edges arranged  in a cycle as shown in Figure~\ref{fig:undirected-gadgets}(b).
	The three labeled vertices correspond to the literal vertices in the clause gadgets.
	The planar embedding of the \textsc{Planar-Monotone-3SAT} instance grantees that we can embed the graph of the directed WSPR instance.
	
	Assume that the \textsc{Planar-Monotone-3SAT} instance have a satisfying assignment.
	We build a weakly simple Euler tour $P$ as follows.
	Each individual gadget defines a cycle.
	As in Algorithm A, we will merge the cycles into the polygon $P$.
	Every cycle will be traversed clockwise, however, cycles defined by variable gadgets are collinear and there is no clear definition of winding direction for them.
    We perturb the red edges based on the truth values of the variables. For each variable assigned \texttt{true} (resp., \texttt{false}),
	we perturb the red edge to pass below (resp., above) the $x$-axis.
	All variable cycles can be safely merged into a single circuit.
	We merge each clause to the variable cycle through a literal vertex of a \texttt{true} variable if the clause is positive or through a literal vertex of a \texttt{false} variable otherwise.
	
	To show that $P$ is weakly simple, we build a simple polygon $P'$ within $\eps$ Fr\'echet distance from $P$ as follows (see Figure~\ref{fig:undirected-hardness}).
	For each $v_i$ create two vertices $v_i^+$ and $v_i^-$ located $\eps/2$ above and below $v_i$ respectively.
	If the solution assigns the $i$-th variable \texttt{true}, move vertices $l_{i,a}$ up by $\eps/2$,
	replace vertices $v_{i-1}$ and $v_i$ by $v_{i-1}^+$ and $v_i^+$ in the black edges (of the corresponding gadget) and by $v_{i-1}^-$ and $v_i^-$ in the red edges.
	Connect vertices $v_0^+$ and $v_0^-$ with an edge. Do the same for $v_n^+$ and $v_n^-$.
	If the variable is assigned \texttt{false}, do analogous replacements symmetrically about the $x$-axis.
	For each clause gadget, choose a literal $l_{i,a}$ with a \texttt{true} value,
	split $l_{i,a}$ into two vertices, $l_{i,a}'$ and $l_{i,a}''$, with the same $y$-coordinate and $\eps$ distance apart,
    such that they each are incident to one edge of the variable gadget and one edge of the clause gadget.
	For the other two literals, split $l_{i,a}$ into two vertices, $l_{i,a}^+$ and $l_{i,a}^-$, with the same $x$-coordinate and $\frac{\eps}{2}$ distance apart,
    such that the one closer to the $x$-axis is incident to two edges of the variable gadget, and the other to two edges of the clause gadget.
	The result is a simple polygon and therefore undirected WSPR have a positive solution.
	
	\begin{figure}[t]
		\centering
		\includegraphics[width=0.45\linewidth]{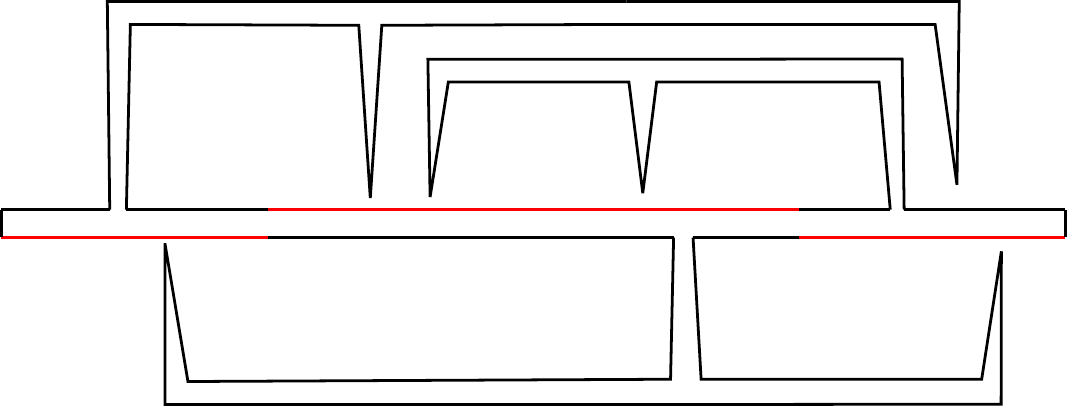}
		\caption{Simple polygon that certifies that an Euler tour of $P$ is weakly simple.}
		\label{fig:undirected-hardness}
	\end{figure}

	Now assume that the graph produced by the reduction admits a weakly simple Euler tour $P$. Then, there exist a simple polygon $P'$ within an arbitrarily small Fr\'echet distance from $P$.
    Such a polygon determines a vertical order between the paths of each variable gadget. Since every literal vertex has degree~4, there are only two possible ways to match its incident edges
    in a noncrossing manner: matching two horizontal edges and two clause edges, or a horizontal with a clause edge. In both cases, the two horizontal black edges incident to a literal vertex are placed above or below the red path. Therefore, all edges of the black path of a variable gadget are on the same side of its red path.
	For each variable, assign \texttt{true} if the black path of its gadget is above the red path and \texttt{false} otherwise.
	Since each clause gadget needs to be connected to some edge in a variable gadget, if the clause is positive/negative, one of its corresponding variables were assigned \texttt{true}/\texttt{false}.
	Hence, the assignment satisfies all clauses and the \textsc{Planar-Monotone-3SAT} instance have a positive solution.
\end{proof}

\begin{lemma}
	\label{lem:directed}
	Directed-WSPR is NP-hard.
\end{lemma}

As a consequence of Lemmas~\ref{lem:np}, \ref{lem:undirected}, and \ref{lem:directed}, we have the following result.

\begin{theorem}
	Both WSPR and Directed-WSPR are NP-complete.
\end{theorem}

\begin{remark}
Our reduction can be modified by perturbing the points in our variable gadgets so that: (i) points belonging to the same gadget are collinear; (ii) no three points, each belonging to a different gadget are collinear; and (iii) no edge crossing is introduced.
By reducing from \textsc{Planar-Monotone-(2,3)-SAT-3}~\cite{DDD16}, in which clauses may have two or three literals and each variable can appear only in up to three clauses, we can show that WSPR remain NP-hard even if the number of mutually collinear points is constant.
\end{remark}

\section{Related problems}
\label{sec:bounds}

Since WSPR is NP-complete in the general case, we study related problems in which a weakly simple polygon is always achievable by
allowing edge subdivision and insertion of new edges.

\subsection{Edge subdivision}
\label{sec:subdivision}


Given a noncrossing graph $G=(V,E)$ where every vertex has even degree and the point set $\bigcup E$ is connected,
we define the problem WSPR$^*$ as finding a sequence of edge subdivision operations that produces a graph $G^*=(V,E^*)$
that admits a weakly simple Euler tour. The \emph{subdivision} of an edge $uv$ at a vertex $w\in {\rm relint}(uv)$
replaces $uv$ by two edges $uw$ and $wv$.

It is easy to see that WSPR$^*$ is always feasible with $O(n^2)$ subdivisions where $n=|V|$. Indeed, subdivide every edge $uv$ recursively at each vertex that lies in the interior of $uv$. We obtain a connected geometric multigraph with even degrees, which admits a weakly simple Euler tour by Corollary~\ref{cor:GeomGraph}.
%
%
The main result of this section is the following.
%

\begin{theorem}\label{thm:subdivision}
Every noncrossing graph $G=(V,E)$ such that every $v\in V$ has even degree and $\bigcup E$ is connected,
can be transformed into a graph $G^*=(V,E^*)$ using $O(|E|)$ edge subdivisions, and this bound cannot be improved.
\end{theorem}

Before the proof, we introduce some notation (from~\cite{AAET16,CEX15,CDP+09}). Let $G=(V,E)$ be a noncrossing graph.
The transitive closure of the \emph{overlap} relation is an equivalence relation on $E$. The union of all edges in an equivalence class is called a \emph{bar}, it is a line segment. A vertex $v\in b$ is called \emph{$b$-odd} if $v$ is incident to an odd number of edges contained in $b$, or \emph{$b$-even} otherwise.
A vertex can be $b$-odd and $b'$-even for different bars $b$ and $b'$ (see Figure~\ref{fig:subdivide}(b)). 

Our algorithm will compute simple paths formed by subdivided edges.
Let $b$ be a horizontal bar with vertices $p_1,p_2\in b$, $x(p_1)\le x(p_2)$. Let $q_1q_2\in E$ be an edge that contains $p_1$ and its right endpoint has minimum $x$-coordinate. A \emph{subdivided paths}, denoted by $\widehat{p_1p_2}$, is a path between $p_1$ and $p_2$, defined recursively (see Fig.~\ref{fig:subdivide}(d)): (i) if $x(p_1)=x(p_2)$, $\widehat{p_1p_2}=\emptyset$;  (ii) if $p_2\in q_1q_2$, then subdivide $q_1q_2$ into three edges $e_1=q_1p_1$, $e_2=p_1p_2$, and $e_3=p_2q_2$ and put $\widehat{p_1p_2}=(e_2)$; (iii)  if $p_2\not\in q_1q_2$, then subdivide $q_1q_2$ into two edges $e_1=q_1p_1$ and $e_2=p_1q_2$, and put $\widehat{p_1p_2}=(e_2)\oplus\widehat{q_2p_2}$, where $\oplus$ denotes concatenation. Consequently, if the segment $p_1p_2$ contains $k$ vertices, a path $\widehat{p_1p_2}$ can be constructed using at most $k$ edge subdivisions.
An example is shown in Figure~\ref{fig:subdivide}(c).

\begin{proof}[Proof of Theorem~\ref{thm:subdivision}]
The proof of the upper bound is constructive. The algorithm subdivides edges within each bar independently. Let $b$ be a bar containing $m$ vertices. We apply $O(m)$ edge subdivisions and partition the edges in $b$ into subsets: Subsets $\mathcal{M}^+$ and $\mathcal{M}^-$ will consists of subdivision paths between the intersection points of $b$ with other bars lying above and below $b$, respectively; all remaining edges will be partitioned into tours (each of which is a weakly simple polygon by Theorem~\ref{thm:collinear}).
The algorithm is divided into three phases: Phase~1 creates $\mathcal{M}^+$ and $\mathcal{M}^-$; phase~2 forms circuits; and phase~3 establishes common vertices between the subdivision paths and circuits. Refer to Figure~\ref{fig:subdivide}.

\noindent{\bf Phase~1.} Compute a list $B^+$ (resp., $B^-$) of points $p$ in the interior of $b$ such that $p$ is $b'$-odd for some bar $b'$ that is above or collinear to $b$ (resp., below $b$). A point can appear more than once in each list if it is odd in multiple bars $b'$.
Sort the lists by $x(p)$, ties are broken by clockwise (resp., counterclockwise) order of the corresponding bars $b'$.
If the left (resp., right) endpoint of $b$ is $b$-odd, add it to the beginning (resp., end) of the list $B^+$.
If any of the lists have odd cardinality, append the right endpoint of $b$ at the end of the list.
Create a perfect matching of consecutive endpoints in each list.
Construct edge disjoint subdivided paths between each pair of matched points, and let $\mathcal{M}^+$ and $\mathcal{M}^-$
denote the set of edges in such paths for $B^+$ and $B^-$, respectively (see Figure~\ref{fig:subdivide}(d)).

\noindent{\bf Phase~2.}
Let $\mathcal{B}$ be the set of (subdivided) edges that lie on $b$ and are not in $\mathcal{M}^+\cup \mathcal{M}^-$.
The union of edges in $\mathcal{B}$ may be a disconnected point set (e.g., as shown in Figure~\ref{fig:subdivide}(d)).
Let the line segment $r_1r_2$ be one of the connected components of $\bigcup\mathcal{B}$.
Construct two edge disjoint subdivided paths $\widehat{r_1r_2}^+$ and $\widehat{r_1r_2}^-$ from the edges in $\mathcal{B}$.
For every path $\widehat{p_1p_2}$ in $\mathcal{M}^+$ (resp., $\mathcal{M}^-$) that overlaps with $r_1r_2$, identify an edge of $\widehat{r_1r_2}^+$ (resp., $\widehat{r_1r_2}^-$)
that contains a vertex of $\widehat{p_1p_2}$ and subdivide it at such vertex (see Figure~\ref{fig:subdivide}(e)).
Let $\mathcal{O}^+$ (resp., $\mathcal{O}^-$) be the set of edges in $\widehat{r_1r_2}^+$ (resp., $\widehat{r_1r_2}^-$) for all components $r_1r_2$ of the union of edges in $\mathcal{B}$.

\noindent{\bf Phase~3.}
Let $\mathcal{B}'$ be the set of edges in $\mathcal{B}\setminus(\mathcal{O}^+\cup\mathcal{O}^-)$.
For every component $C$ of the subgraph induced by $\mathcal{B}'$, let $p_0$ be the leftmost vertex of $C$,
identify the edge in $\mathcal{O}^+$ that contains $p_0$ and subdivide it at $p_0$.
This concludes the construction of $G^*$.

\begin{figure}[h]
	\centering
	\def\svgwidth{.9\columnwidth}
\begingroup%
  \makeatletter%
  \providecommand\color[2][]{%
    \errmessage{(Inkscape) Color is used for the text in Inkscape, but the package 'color.sty' is not loaded}%
    \renewcommand\color[2][]{}%
  }%
  \providecommand\transparent[1]{%
    \errmessage{(Inkscape) Transparency is used (non-zero) for the text in Inkscape, but the package 'transparent.sty' is not loaded}%
    \renewcommand\transparent[1]{}%
  }%
  \providecommand\rotatebox[2]{#2}%
  \ifx\svgwidth\undefined%
    \setlength{\unitlength}{506.4722785bp}%
    \ifx\svgscale\undefined%
      \relax%
    \else%
      \setlength{\unitlength}{\unitlength * \real{\svgscale}}%
    \fi%
  \else%
    \setlength{\unitlength}{\svgwidth}%
  \fi%
  \global\let\svgwidth\undefined%
  \global\let\svgscale\undefined%
  \makeatother%
  \begin{picture}(1,0.35616639)%
    \put(0,0){\includegraphics[width=\unitlength,page=1]{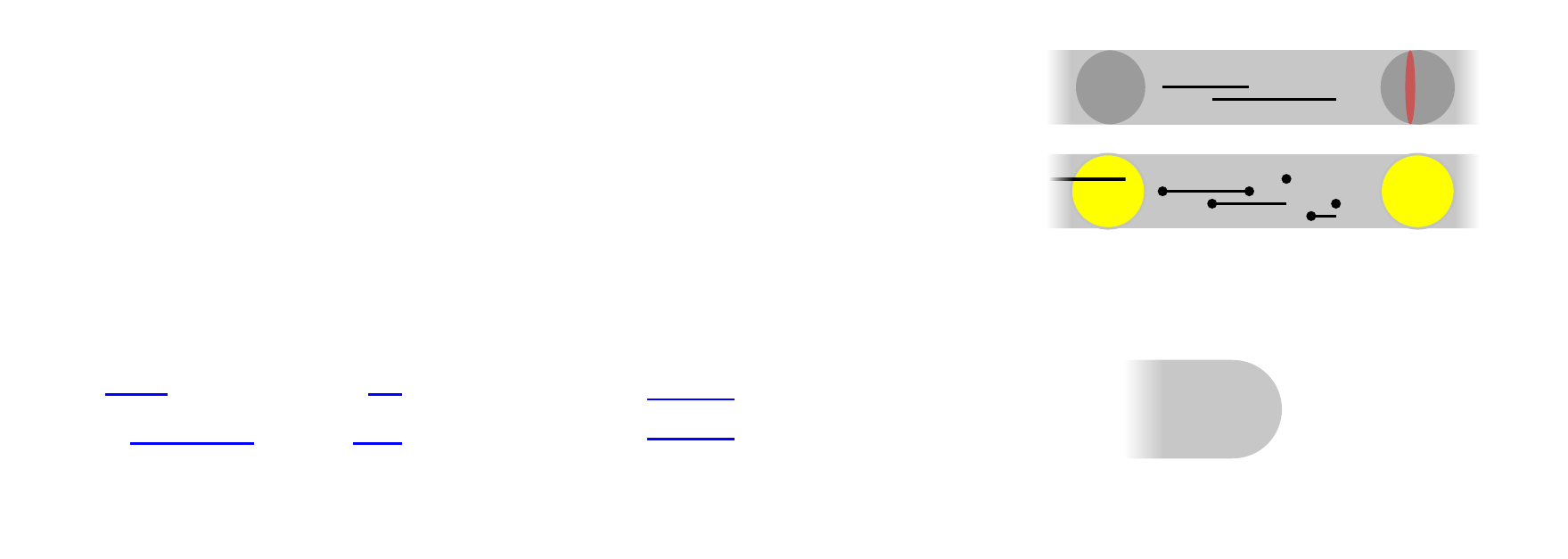}}%
    \put(0.70009773,0.33663222){\color[rgb]{0,0,0}\makebox(0,0)[lb]{\smash{$p_1$}}}%
    \put(0.88851584,0.33663222){\color[rgb]{0,0,0}\makebox(0,0)[lb]{\smash{$p_2$}}}%
    \put(0,0){\includegraphics[width=\unitlength,page=2]{subdivide-eg.pdf}}%
    \put(0.7926301,0.16649171){\color[rgb]{0,0,0}\makebox(0,0)[lb]{\smash{(c)}}}%
    \put(0,0){\includegraphics[width=\unitlength,page=3]{subdivide-eg.pdf}}%
    \put(0.16553171,0.20237017){\color[rgb]{0,0,0}\makebox(0,0)[lb]{\smash{(a)}}}%
    \put(0.49681607,0.20237017){\color[rgb]{0,0,0}\makebox(0,0)[lb]{\smash{(b)}}}%
    \put(0.13564871,0.0065056){\color[rgb]{0,0,0}\makebox(0,0)[lb]{\smash{(d)}}}%
    \put(0.69109117,0.04777741){\color[rgb]{0,0,0}\makebox(0,0)[lb]{\smash{}}}%
    \put(0,0){\includegraphics[width=\unitlength,page=4]{subdivide-eg.pdf}}%
    \put(0.24560625,0.12324652){\color[rgb]{0,0,0}\makebox(0,0)[lb]{\smash{$\mathcal{M^+}$}}}%
    \put(0.24669446,0.03520075){\color[rgb]{0,0,0}\makebox(0,0)[lb]{\smash{$\mathcal{M^-}$}}}%
    \put(0,0){\includegraphics[width=\unitlength,page=5]{subdivide-eg.pdf}}%
    \put(0.48157084,0.0065056){\color[rgb]{0,0,0}\makebox(0,0)[lb]{\smash{(e)}}}%
    \put(0,0){\includegraphics[width=\unitlength,page=6]{subdivide-eg.pdf}}%
    \put(0.60836237,0.10138734){\color[rgb]{0,0,0}\makebox(0,0)[lb]{\smash{$\mathcal{O^+}$}}}%
    \put(0.60836237,0.0614976){\color[rgb]{0,0,0}\makebox(0,0)[lb]{\smash{$\mathcal{O^-}$}}}%
    \put(0,0){\includegraphics[width=\unitlength,page=7]{subdivide-eg.pdf}}%
    \put(0.74702552,0.02999976){\color[rgb]{0,0,0}\makebox(0,0)[lb]{\smash{$D_{u_i}$}}}%
    \put(0.81058055,0.00510666){\color[rgb]{0,0,0}\makebox(0,0)[lb]{\smash{(f)}}}%
    \put(0,0){\includegraphics[width=\unitlength,page=8]{subdivide-eg.pdf}}%
  \end{picture}%
\endgroup%
	\caption{(a) A bar $b$ and its adjacent bars. (b) Each bar $b'$ is shown with its $b'$-odd and $b'$-even vertices shown in red and green respectively. (c) The subdivided path $\widehat{p_ip_2}$ is shown in blue. Examples of (d) $\mathcal{M^+}$ and $\mathcal{M^-}$; (e) $\mathcal{O^+}$ and $\mathcal{O^-}$. (f) Connecting a component of $\mathcal{B}'$ to a path in $\mathcal{O^+}$ with two polygonal paths shown in magenta.}
	\label{fig:subdivide}
\end{figure}

\noindent{\bf Correctness.}
Now we prove that $G^*$ admits a weakly simple Euler tour.
Notice that $G^*$ is connected since the subdivisions in phase~2 connects every component of $\mathcal{M^+}$ or $\mathcal{M^-}$ to every overlapping component of $\mathcal{O}^+$ or $\mathcal{O}^-$, and phase~3 connects every component of $\mathcal{B}'$ to some component in $\mathcal{O}^+$.
Since edge subdivisions do not change the parity of degrees, every vertex in $G^*$ has even degree, hence $G^*$ is Eulerian.
We construct an Eulerian geometric graph $G'$ such that every Euler tour in $G'$ is within $\eps/2$ Fr\'echet distance from an Euler tour in $G^*$. Theorem~\ref{thm:GeomGraph} will then imply that there exists a simple polygon within $\eps$ Fr\'echet distance from an Euler tour in $G^*$.

We recall some notation introduced in \cite{CEX15}. For every vertex $v\in V$, let $D_v$ be a disk centered at $v$ of radius $\frac{\eps}{4}$. For a bar $b$ between $u_0$ and $u_k$, let $D_b$ be the $\eps^2$ neighborhood of $b$ setminus $D_{u_0}\cup D_{u_k}$. Assume that $\eps\in (0,\frac{1}{4})$ is so small that the disks $D_v$ are pairwise disjoint; a disk $D_v$ intersects $D_b$ only if $v\in b$, and the neighborhoods $D_b$ are pairwise disjoint.

For each bar $b$ with vertices $u_0,\ldots , u_k$, we perturb the edges of $E^*$ contained in $b$ into noncrossing simple polygons and polygonal chains. Embed each subdivided path in $\mathcal{M}^+$ (resp., $\mathcal{M}^-$) $\widehat{u_iu_j}$ in the upper (resp., lower) boundary of the region $D_b$ such that $u_i$ is on the boundary of $D_{u_i}$ and $u_j$ is on the boundary of $D_{u_j}$.
Subdivide $D_b$ with $\ell+1$ horizontal lines where $\ell$ is the number of components of the subgraph induced by $\mathcal{B}'$.
Embed all edges in $\mathcal{O}^+$ (resp., $\mathcal{O}^-$) in the first (resp., $\ell+1$-th) such line.

Recall that every vertex of $G$ has even degree.
If $b$ contains a $b$-odd vertex $p$, there must exist a bar $b'$ such that $p$ is $b'$-odd.
Because $\mathcal{M}^+$ and $\mathcal{M}^-$ matches such points, the subgraph induced by $\mathcal{B}$ contain only even degree vertices.
By construction the edges in $\mathcal{O}^+\cup \mathcal{O}^-$ form nonoverlapping disjoint circuits.
Hence, the subgraph induced by $\mathcal{B}'$ contains only even degree vertices. Consequently, each of its components is Eulerian and forms a weakly simple polygons that we denote by $\gamma_1(b),\ldots ,\gamma_{\ell}(b)$, sorted by the $x$-coordinates of their left endpoints.
Perturb $\gamma_1(b),\ldots ,\gamma_{\ell}(b)$ into simple polygons that lie in the interior of $D_b$, separated by one of the $\ell+1$ lines, in this linear order (ties are broken arbitrarily).
For $i=0,\ldots , k$, consider all polygons $\gamma_j(b)$ whose leftmost vertex is $u_i$. Connect the left endpoints of each such $\gamma_j(b)$
to the copy of $u_i$ in $\mathcal{O^+}$ by two polygonal paths within $D_{u_i}$ (these paths connect different copies of vertex $u_i\in V^*$, see Figure~\ref{fig:subdivide}(f)).
Similarly, for each subdivision performed in phase 2 at $u_i$ of an edge in $\mathcal{O^+}$ (resp., $\mathcal{O^-}$) in a path $\widehat{r_1r_2}$, connect the copy of $u_i$ in this path to a copy in $\mathcal{M}^+$ (resp., $\mathcal{O^-}$) by two polygonal paths within the disk $D_{u_i}$.
Connect the endpoints of the overlapping paths in $\mathcal{O^+}$ and $\mathcal{O^-}$ (forming a cycle of each), and if the right endpoint of $b$ is not $b$-odd and was added to $B^+, B^-$, connect the copies of $u_k$ in $\mathcal{M}^+$ and $\mathcal{M}^-$.
For each matching in $B^+$ and $B^-$ involving a point $u_i$ that is a $b'$-odd endpoint of a bar $b'$, connect the paths in $\mathcal{M}^+$ or  $\mathcal{M}^-$ that correspond to a match in $b$ to the path in $\mathcal{M}^+$ of $b'$ that contains $p$.
Finally, for each point $u_i$ that is the endpoint of a bar $b'$ and is $b'$-even, connect the corresponding copies of $u_i$, making the graph induced by all edges containing a point on $b$ connected. This concludes the construction of $G'$.

Theorem~\ref{thm:GeomGraph} completes the proof: An Euler tour $\widehat{P}$ of $G'$ can be perturbed into a simple polygon $P$ such that $\dist_F(P,\widehat{P})< \frac{\eps}{2}$.
The tour $\widehat{P}$ maps to an Euler tour $P^*$ of $G^*$ by identifying the vertices that lie in the same disk $D_v$, $v\in V^*$; and $\dist_F(\widehat{P},P^*)< \frac{\eps}{2}$.

Our lower bound construction is shown in Figure~\ref{fig:lowerbounds}(a).
It consists of a graph $G=(V,E)$ containing a long edge $e_R$ (shown in red) and a path of $(|E|+5)/7$ non-overlapping collinear edges that connects the endpoints of $e_R$.
Each vertex in the interior of $e_R$ is also incident to two small cycles above and below $e_R$ respectively.
Although the graph is Eulerian, it does not admit a weakly simple Euler tour.
Each vertex $p$ in the interior of the red edge $e_R$ is incident to two small triangles.
Suppose that $e_R$ is \emph{not} subdivided at $p$. Then $p$ has degree 6.
In any perturbation of a weakly simple Euler tour, vertex $p$ is split into 3 copies,
each of degree 2, and each lying above or below $e_R$. Suppose only one copy of $p$
lies below $e_R$. Then it is incident to two edges of a small triangle below $e_R$,
which is then disconnected from the rest of the graph, a contradiction.
Consequently, $e_R$ must be subdivided at all interior vertices.
Figure~\ref{fig:lowerbounds}(b) shows that $O(|E|)$ subdivisions of $e_R$ suffice
in this case.
\end{proof}

\begin{figure}[b]
	\centering
	\def\svgwidth{0.9\columnwidth}
\begingroup%
  \makeatletter%
  \providecommand\color[2][]{%
    \errmessage{(Inkscape) Color is used for the text in Inkscape, but the package 'color.sty' is not loaded}%
    \renewcommand\color[2][]{}%
  }%
  \providecommand\transparent[1]{%
    \errmessage{(Inkscape) Transparency is used (non-zero) for the text in Inkscape, but the package 'transparent.sty' is not loaded}%
    \renewcommand\transparent[1]{}%
  }%
  \providecommand\rotatebox[2]{#2}%
  \ifx\svgwidth\undefined%
    \setlength{\unitlength}{371.5756538bp}%
    \ifx\svgscale\undefined%
      \relax%
    \else%
      \setlength{\unitlength}{\unitlength * \real{\svgscale}}%
    \fi%
  \else%
    \setlength{\unitlength}{\svgwidth}%
  \fi%
  \global\let\svgwidth\undefined%
  \global\let\svgscale\undefined%
  \makeatother%
  \begin{picture}(1,0.29361613)%
    \put(0,0){\includegraphics[width=\unitlength,page=1]{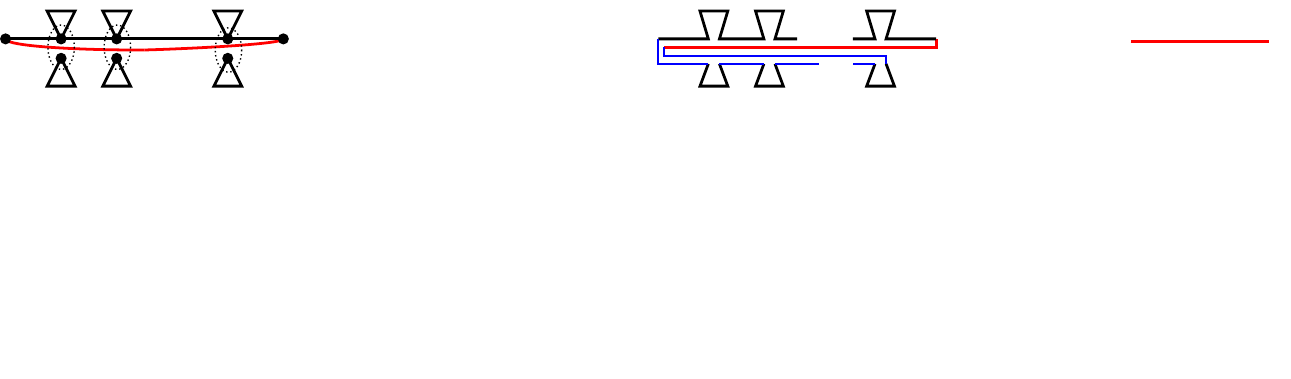}}%
    \put(0.09326117,0.19820666){\color[rgb]{0,0,0}\makebox(0,0)[lb]{\smash{(a)}}}%
    \put(0.34593345,0.19875037){\color[rgb]{0,0,0}\makebox(0,0)[lb]{\smash{(b)}}}%
    \put(0,0){\includegraphics[width=\unitlength,page=2]{lowerbounds.pdf}}%
    \put(0.60429262,0.19875073){\color[rgb]{0,0,0}\makebox(0,0)[lb]{\smash{(c)}}}%
    \put(0,0){\includegraphics[width=\unitlength,page=3]{lowerbounds.pdf}}%
    \put(0.86419176,0.19880227){\color[rgb]{0,0,0}\makebox(0,0)[lb]{\smash{(d)}}}%
    \put(0.10674468,0.06971362){\color[rgb]{0,0,0}\makebox(0,0)[lb]{\smash{(e)}}}%
    \put(0.45492038,0.00713398){\color[rgb]{0,0,0}\makebox(0,0)[lb]{\smash{(f)}}}%
    \put(0.80093721,0.00713398){\color[rgb]{0,0,0}\makebox(0,0)[lb]{\smash{(g)}}}%
    \put(0,0){\includegraphics[width=\unitlength,page=4]{lowerbounds.pdf}}%
  \end{picture}%
\endgroup%
	\caption{Lower bound constructions.}
	\label{fig:lowerbounds}
\end{figure}

\subsection{Edge insertion}
\label{sec:insertion}

We define the problem WSPR$^+$ as finding a set of edges $E^+$ such that
$G^+=(V,E\cup E^+)$ admits a weakly simple Euler tour.
Denote by $\|E\|$ and $\|E^+\|$, respectively, the sum of the lengths of all edges in $E$ and $E^+$.
If the point set $\bigcup E$ is disconnected, then there is no upper bound on $\|E^+\|$.
Otherwise, we can establish worst-case upper and lower bounds for $\|E^+\|$ in terms of $\|E\|$.
%

\begin{theorem}\label{thm:upperbound}
Let $G=(V,E)$ be a noncrossing multigraph such that $\bigcup E$ is a connected point set.
Then there exists a set of line segments $E^+$ such that
$\|E^+\|\leq 3\|E\|$ and $G^+=(V,E\cup E^+)$ admits a weakly simple Euler tour.
\end{theorem}
\begin{proof}
We construct $E^+$ as follows. Partition $E$ into \emph{bars} (equivalence classes of the transitive closure of the overlap relation on $E$). Denote by $b$ the union of edges in a bar. W.l.o.g., we may assume that $b$ is horizontal. Denote by $u_0,\ldots ,u_k\in V$ the vertices of $V$ along $b$ sorted by $x$-coordinates (where $b=u_0u_k$). For $i=1,\ldots , k$, add an edge $u_{i-1}u_i$ to $E^+$ if the edges of $E$ in the bar cover the line segment $u_{i-1}u_i$ an odd number of times. The old and new edges in the bar $b$ jointly form a graph of even degree that we denote by $G(b)$. By Theorem~\ref{thm:collinear}, every component of $G(b)$ admits a weakly simple Euler tour. Finally, add two more copies of edge $u_{i-1}u_i$ to $E^+$ for all $i=1,\ldots k$. After repeating the above steps for every bar, we have $\|E^+\|\leq 3\|E\|$ and $G^+=(V,E\cup E^+)$ is Eulerian.

We omit the proof of correctness (which is provided in Appendix A.), that shows that $G^+$ admits a weakly simple Euler tour, since it is similar to the proof of Theorem~\ref{thm:subdivision}.
\end{proof}

\noindent{\bf Lower bound constructions.}
All our lower bound constructions are graphs $G=(V,E)$ in which an edge connects two points on the boundary of the convex hull of $V$,
denoted ${\rm ch}(V)$.



\begin{theorem}\label{thm:lowerbounds}
Let $\mathcal{G}$ be a family of noncrossing multigraphs. For $G=(V,E)\in \mathcal{G}$, let $E^+$ be an edge set of minimum length $\|E^+\|$ such that $G^+=(V,E\cup E^+)$ admits a weakly simple Euler tour; and let $\lambda(\mathcal{G}) = \sup_{G\in \mathcal{G}} \|E^+\|/\|E\|$. Then:
\begin{enumerate}
\item $\lambda(\mathcal{G}_1)\geq 1$, where $\mathcal{G}_1=\{$Eulerian noncrossing multigraphs$\}$.
\item $\lambda(\mathcal{G}_2)\geq \frac65$, where $\mathcal{G}_2=\{$connected noncrossing multigraphs$\}$.
\item $\lambda(\mathcal{G}_3)\geq 3$, where $\mathcal{G}_3=\{$noncrossing multigraphs $G=(V,E)$ such that $\bigcup E$ is connected$\}$.
\end{enumerate}
\end{theorem}
\begin{proof}
{\bf (1)} Refer to Figs.~\ref{fig:lowerbounds}(a)--(c).
Let $n\in \mathbb{N}$ and $\delta\in (0,\frac13)$.
Place vertices $v_i=(i,0)$, for $i=0,\ldots , n$, on the $x$-axis.
A red edge of length $n$ connects $v_0$ and $v_n$.
A black edge of length $1/n$ connects $v_{i-1}$ and $v_i$ for $1\leq i\leq n$.
A small cycle of length $\delta<\frac13$ is placed on each $v_i$, $1\leq i\in n-1$, on each side of the $x$-axis.
The total length of the construction is $\|E\|=2n+2(n-1)\delta$.
	
Let $G^+=(V,E\cup E^+)$ be a multigraph in which $P$ is a weakly simple Euler tour;
and let $P'$ be an $\eps$-perturbation into a simple polygon, for some $0<\eps<\delta$.
We define a pair of vertical lines $\ell_i^-:x=i+\delta$ and $\ell_i^+:x=(i+1)-\delta$,
for $0\leq i\leq n-1$. The portion of $P'$ between any two of these lines consists of disjoint
paths whose endpoints are on the lines. By Morse theory, $P'$ contains an even number of paths between any
two of these lines; and the length of such a path is at least the distance between the parallel lines.
The input edges already contain two line segments between any two of these lines: a red and a black segment.
	
We claim that $G^+$ contains at least 4 paths between $\ell_i^-$ and $\ell_i^+$ for all but at most one index $0\leq i\leq n-1$.
Indeed, suppose that there are two such paths between $\ell_i^-$ and $\ell_i^+$ and between $\ell_j^-$ and $\ell_j^+$ ($0\leq i <j\leq n-1$).
We may assume w.l.o.g. that the black edge is above the red edge between $\ell_i^-$ and $\ell_i^+$.
Then the black edge must be above the red edge between $\ell_j^-$ and $\ell_j^+$, as well.
Consequently, $P'$ cannot reach the small cycles at $v_{i+1},\ldots , v_j$. This confirms the claim.
It follows that $\|E\cup E^+\|\geq (4n-2)(1-2\delta)$. This lower bound tends to $2\|E\|$ as
$n\rightarrow \infty$ and $\delta n\rightarrow 0$.

Due to space restrictions, we omit the proofs for cases 2 and 3.
\end{proof}

\section{Conclusions}
\label{sec:conclusions}

We have shown that WSPR is NP-complete.
It follows that the decision version of the problems in Section~\ref{sec:bounds} are also NP-complete:
It is NP-complete to find up to $k$ subdivision points to form a weakly simple polygon, or to find an edge set with length up to $k$ that produce a weakly simple polygon.
We have shown that $\Theta(|E|)$ subdivision points are always sufficient and sometimes necessary when the input is Eulerian; and new edges of length $\Theta(\|E\|)$ are always sufficient and sometimes necessary when $\bigcup E$ is connected. However, the best constant coefficients are not known in most cases. We conjecture that every noncrossing Eulerian graph $G=(V,E)$ can be augmented into a graph $G^+=(V,E\cup E^+)$ that admits a weakly simple Euler tour such that $\|E^+\|\leq \|E\|$.

If the segments in $E$ do not form a weakly simple polygon, we can subdivide segments or insert new segments to create a weakly simple polygon. On the other end of the spectrum, a set of $n$ line segment may form an exponential number of weakly simple polygons, even if all segments are collinear.
It is an open problem to count exactly how many weakly simple polygons can be obtained from the same set of line segments.
Finally, we mention an open problem about reconstructing \emph{simple polygons} from a subset of its edges. It is NP-complete to decide whether a geometric graph $G=(V,E)$ can be augmented into a \emph{simple} polygon $P=(V,E\cup E^+)$~\cite{Rap89}. However, it is not known whether the problem remains NP-hard when $G$ is a perfect matching.

\smallskip\noindent{\bf Acknowledgements.} Research on this paper was partially supported by NSF awards CCF-1422311 and CCF-1423615. The first author was supported by the Science Without Borders program. We thank Adrian Dumitrescu for bringing this problem to our attention.





\newpage
\appendix

\section{Omitted proofs}

\begin{proof}[Proof of Lemma 6]
	The reduction to Directed-WSPR is very similar to the one described in the proof of Lemma~\ref{lem:undirected}.
	Figure~\ref{fig:directed-hardness}(a) shows an example of the reduction with the edges on the variable line distorted to show the overlapping segments.
	Black disks circled by dotted ellipses represent the same vertex. The key difference is the addition of four segments colored blue in Figure~\ref{fig:directed-hardness}.
	
	\begin{figure}[h]
		\centering
		\includegraphics[width=0.6\linewidth]{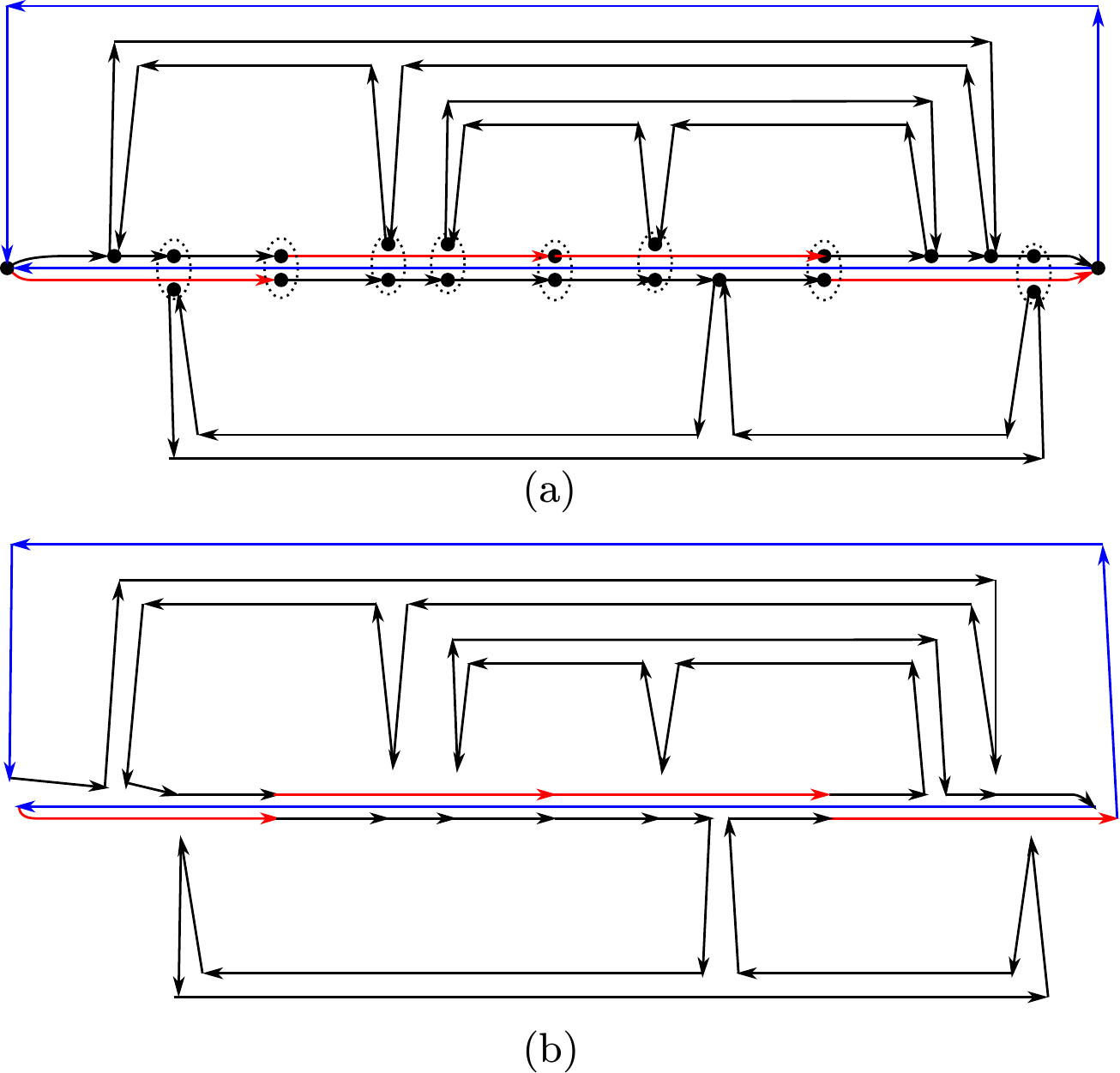}
		\caption{(Left) Reduction from \textsc{Planar-Monotone-3SAT} to \textsc{Directed-WS} and (right) a possible solution. }
		\label{fig:directed-hardness}
	\end{figure}
	
	All segments of the variable gadgets are directed to the right and the clause gadgets form clockwise cycles.
	If the \textsc{Planar-Monotone-3SAT} have a positive solution, we can construct a weakly simple polygon as shown in Figure~\ref{fig:directed-hardness}(b) in a similar way as described in the proof of Lemma~\ref{lem:undirected}.
	The difference is that we place a blue edge in between the black and red paths of the variable gadgets and connect $v_n$ and $v_0$ with a blue path above the construction.
	Analogously to the proof of Lemma~\ref{lem:undirected}, if there is a positive solution for the Directed-WSPR instance, all edges of the black path of a variable gadget are on the same side of the red path.
	Then we can assign \texttt{true}/\texttt{false} values as previously explained, and obtain a satisfying assignment.
\end{proof}

\begin{proof}[Proof of correctness (Theorem~\ref{thm:upperbound})]
	Here we show that $G^+$ admits a weakly simple Euler tour. Note that an $\frac{\eps}{2}$-perturbation of an $\frac{\eps}{2}$-perturbation is an $\eps$-perturbation.
	We shall construct an $\frac{\eps}{2}$-perturbation of $G^+$ into an Eulerian \emph{geometric graph}, and then Theorem~\ref{thm:GeomGraph} yields and $\frac{\eps}{2}$-perturbation into a simple polygon for all $\eps>0$.
	
	For each bar $b$ with vertices $u_0,\ldots , u_k$, we perturb the edges of $E\cup E^+$ contained in $b$ into simple polygons: Let $\beta(b)$ be the circuit formed by two copies of the path $(u_0,\ldots, u_k)$, and embed $\beta(b)$ on the boundary of the region $D_b$ such that $u_0$ is on the boundary of $D_{u_0}$ and $u_k$ is on the boundary of $D_{u_k}$. The components of $G(b)$ each form weakly simple polygons that we denote by $\gamma_1(b),\ldots ,\gamma_{\ell(b)}$, sorted by the $x$-coordinates of their left endpoints. Perturb $\gamma_1(b),\ldots ,\gamma_{\ell(b)}(b)$ into simple polygons that lie in the interior of $D_b$, separated by lines parallel to $b$, in this linear order (ties are broken arbitrarily). For $i=0,\ldots , k$, consider all polygons $\gamma_j(b)$ whose leftmost vertex if $u_i$. 
Connect the left endpoints of each such $\gamma_j(b)$
to the copy of $u_i$ in the upper arc of $\beta(b)$ by two polygonal paths within $D_{u_i}$ (these paths connect different copies of vertex $u_i\in V$). Similarly, if bars $b$ and $b'$ intersect at some vertex $v=u_i=u_j'$, then connect the corresponding vertices in $\beta(b)$ and $\beta(b')$ by two polygonal paths within the disk $D_v$. These connector paths augment the pairwise disjoint polygons $\beta(b)$, $\gamma_1(b),\ldots, \gamma_{\ell(b)}(b)$, for all bars $b$, into a simple Eulerian noncrossing geometric graph $\widehat{G}$.
	
	Theorem~\ref{thm:GeomGraph} completes the proof: An Euler tour $\widehat{P}$ of $\widehat{G}$ can be perturbed into a simple polygon $P$ such that $\dist_F(P,\widehat{P})< \frac{\eps}{2}$.
	The Euler tour $\widehat{P}$ of $\widehat{G}$ maps to an Euler tour $P^+$ of $G^+$ by identifying the vertices that lie in the same disk $D_v$, $v\in V$; and $\dist_F(\widehat{P},P^+)< \frac{\eps}{2}$.
\end{proof}

The following lemma provides a canonical form of a weakly simple Euler tour when an edge connects two points on the boundary of the convex hull of $V$.

\begin{lemma}\label{lem:3gamma}
	Let $G=(V,E)$ be a multigraph that admits a weakly simple Euler tour and an edge $ab\in E$ where $u$ and $v$ are on the boundary of ${\rm ch}(V)$. Then $G$ admits a weakly simple Euler tour composed of 3 weakly simple polygons: $\gamma_1$ contains all edges that lie in one of the open halfplanes bounded by $uv$;
	$\gamma_2\cup \gamma_3$ contains all edges in the other open halfplane, $\gamma_2$ is incident to $u$, and $\gamma_3$ is incident to $v$.
\end{lemma}
\begin{proof}
	Assume w.l.o.g. that $ab$ is horizontal. Let $P$ be an arbitrary weakly simple Euler tour of $G$. Cut $P$ at the vertices $u$ and $v$ into paths $\mathcal{P}=\{P_1,\ldots ,P_k\}$.
	Note that every path in $\mathcal{P}$ consists of edges that lie in one closed halfplane bounded by $uv$, since every subpath of $P$ between the two open halfplanes must go through $u$ or $v$.
	
	Clearly, the number of paths in $\mathcal{P}$ between $u$ and $v$ is even, and the path $(u,v)\in \mathcal{P}$ lies in both closed halfplanes. Therefore, there is a partition $\mathcal{P}=\mathcal{P}^-\cup \mathcal{P}^+$ such that every path in $\mathcal{P}^-$ (resp., $\mathcal{P}^+$) lies in the closed halflane below (resp., above) $uv$; and $\mathcal{P}^-$ and $\mathcal{P}^+$ each contain an even number of paths between $u$ and $v$. Assume w.l.o.g. that $(u,v)\in \mathcal{P}^+$. Then the paths in $\mathcal{P}^+$ form a weakly simple polygon. The paths in $\mathcal{P}^-$ incident to $u$ also form a weakly simple polygon; and the remaining paths in $\mathcal{P}^-$ (all incident to $v$ only) also form a weakly simple polygon, as claimed.
\end{proof}

\begin{proof}[Proof of cases (2) and (3) (Theorem~\ref{thm:lowerbounds})]
\smallskip\noindent{\bf (3)}
The construction is similar to case (1) above. Refer to Figs.~\ref{fig:lowerbounds}(d)--(e).
Let $n\in \mathbb{N}$ and $\delta\in (0,\frac13)$.
Place vertices $v_i=(i,0)$, for $i=0,\ldots , n$, on the $x$-axis.
A red edge of length $n$ connects $v_0$ and $v_n$.
A small vertical edge of length $\delta$ is attached to each $v_i$, $1\leq i\in n-1$, on each side of the $x$-axis.
The total length of the construction is $\|E\|=n+2(n-1)\delta$.

Let $G^+=(V,E\cup E^+)$ be a multigraph in which $P$ is a weakly simple Euler tour;
and let $P'$ be an $\eps$-perturbation into a simple polygon, for some $0<\eps<\delta$.
Consider the vertical lines $\ell_i^-:x=i+\delta$ and $\ell_i^+:x=(i+1)-\delta$, for $0\leq i\leq n-1$, defined above.
Analogously to case (1), $G^+$ contains at least 4 path between $\ell_i^-$ and $\ell_i^+$ for all but
at most one index $0\leq i\leq n-1$. It follows that $\|E\cup E^+\|\geq (4n-2)(1-2\delta)$.
This lower bound tends to $4\|E\|$ as $n\rightarrow \infty$ and $\delta n\rightarrow 0$.

\smallskip\noindent{\bf (2)}
Refer to Figs.~\ref{fig:lowerbounds}(f)--(g).  Let $n\in \mathbb{N}$ be even, and $\delta\in (0,\frac13)$.
Place vertices $v_i=(i,0)$, for $i=0,\ldots , 2n$, on the $x$-axis. For $i=0,\ldots , 2n$, let
$v^+_i=(i,\delta)$, $v^{++}_i=(i,2\delta)$,  $v^-_i=(i,-\delta)$, and $v^{--}_i=(i,-2\delta)$.
Vertices $v_0$ and $v_{2n}$ are connected by a red edge $v_0v_n$ of length $2n$, and a black path
$(v_0,v^+_0,v^+_2,v_2,v^-_2,v^-_4,v^-_6,\ldots v_{2n}$, that alternated between the two sized of the red edge.
Finally, for every $i=1,\ldots ,n$, we have a path $(v_{2i}^+,v_{2i}^{++},v_{2i-1}^{++}, v_{2i-1}^{+}, v_{2i-1}, v_{2i-1}^-$
when $i$ is odd, and a reflection of this path in the $x$-axis when $i$ is even.
The total length of the construction is $\|E\|=5n+6n\delta$.

Let $G^+=(V,E\cup E^+)$ be a multigraph in which $P$ is a weakly simple Euler tour.
The red edge connects two points on the boundary of ${\rm ch}(V)$, we can apply Lemma~\ref{lem:3gamma}, and
$P$ can be decomposed into three weakly simple tours: W.l.o.g., $\gamma_1$ contains the red edge and all edges of $P$ above
the $x$-axis; $\gamma_2$ (resp., $\gamma_3$) is incident to $v_0$ (resp., $v_{2n}$) and contains edges below the $x$-axis.

Let us consider the edges of $E^+\cap \gamma_1$. For every odd integer $1\leq i\leq n-1$, the edge $v_{2i-1}v_{2i-1}^+$ is in $\gamma_1$,
its two adjacent edges are in $E^+\cap \gamma_1$, and their length between the lines $x=2i-2$ and $x=2i$ is at least \boxed{2}.
Similarly, edge $v_{2i-1}^+v_{2i-1}^{++}$ is in $\gamma_1$, at least one of its adjacent edges is in $E^+\cap \gamma_1$,
and its length between the lines $x=2i-2$ and $x=2i$ is at least \boxed{1}. Every vertical line crosses the edges of $\gamma_1$ an even number of times, there are edges in $E^+\cap \gamma_1$ of length at least \boxed{2} between the lines $x=2i+2$ and $x=2i+4$.
Summation over all odd indices $i=1,\ldots ,n-1$ yields $\|E^+\cap \gamma_1\|\geq 5\cdot\frac{n}{2}$.

Consider now $E^+\cap (\gamma_2\cup \gamma_3)$. For every odd integer $1\leq i\leq n-1$, the edge $v_{2i+1}v_{2i+1}^-$ is in $\gamma_2\cup \gamma_3$, its two adjacent edges are in $E^+\cap (\gamma_2\cup \gamma_3)$, and their length between the lines $x=2i$ and $x=2i+2$ is at least \boxed{2}.
Similarly, edge $v_{2i+1}^-v_{2i+1}^{--}$ is in $\gamma_2\cup \gamma_3$, at least one of its adjacent edges is in $E^-\cap (\gamma_2\cup \gamma_3)$, and its length between the lines $x=2i$ and $x=2i+2$ is at least \boxed{1}. If a vertical line crosses $\gamma_2\cup \gamma_3$, it must cross it an even number of times. There are edges in $E^+\cap( \gamma_1$ of length at least \boxed{4} between the lines $x=2i-2$ and $x=2i$ for all but at most one odd index $i=1,\ldots n-1$. Summation over all odd indices $1\leq i\leq n-1$ yields $\|E^+\cap (\gamma_2\cap \gamma_3)\|\geq 7\cdot \frac{n}{2}-2$.
Overall, we have $\|E^+\|\geq (5+7)\frac{n}{2}-2=6n-2$. This lower bound tends to $\frac65 \|E\|$ as $n\rightarrow \infty$ and $\delta n\rightarrow 0$.
\end{proof}

\end{document}